\documentclass[algorithms,conceptpaper,accept,moreauthors,pdftex,12pt,a4paper]{mdpi} 
\lastpage{13}
\doinum{10.3390/------}
\pubvolume{8}
\pubyear{2015}
\externaleditor{Academic Editor: Javier Del Ser Lorente}
\history{Received: 18 June 2015 / Accepted: 19 August 2015 / Published: xx}

\usepackage{setspace}
\usepackage{graphicx}
\usepackage{subfigure}
\usepackage{url}
\usepackage{verbatim}
\usepackage{algorithm}
\usepackage{algorithmic}
\usepackage{array}
\usepackage{amsmath}
\newtheorem{theorem}{Theorem}
\usepackage{amssymb}
\usepackage{latexsym}
\usepackage{algorithm,algorithmic}
\usepackage{lmodern}
\usepackage{soul}
 \usepackage{booktabs}
 \usepackage{times}
\usepackage{mathptmx}
\theoremstyle{mdpi}

\Title{Network Community Detection on Metric Space}

\Author{Suman Saha* and Satya P. Ghrera }
\address[1]{%
Dept. Computer Science and Engineering, Jaypee University of Information Technology, Waknaghat, Solan 173215, Himachal, India; E-Mail: suman.saha@juit.ac.in, sp.ghrera@juit.ac.in} 

\corres{E-Mail: suman.saha@juit.ac.in; Tel: +917602728183}

\abstract{Community detection in a complex network is an important problem of much interest in recent years. In general, a community detection algorithm chooses an objective function and captures the communities of the network by optimizing the objective function, and then, one uses various heuristics to solve the optimization problem to extract the interesting communities for the user.
In this article, we demonstrate the procedure to transform a graph into points of a metric space and develop the methods of community detection with the help of a metric defined for a pair of points. We have also studied and analyzed the community structure of the network therein. 
The results obtained with our approach are very competitive with most of the well-known algorithms in the literature, and this is justified over the large collection of datasets. On the other hand, it can be observed that time taken by our algorithm is quite less compared to other methods and justifies the theoretical findings.
}


\begin{document}
\vspace{-12pt}

\section{Introduction}

The rise of on-line networking communities in real-world graphs, such as large social networks, web graphs and biological networks, have initiated the important direction of network community \mbox{detection 
\cite{Freeman78centralityin, CarringtonScott:2005, newman2003structure, radicchi2004}}. 
A network community (also known as a module or cluster) is typically a group of nodes with more interactions among its members than the remaining part of the \mbox{network \cite{weiss, Schaeffer200727, Fortunato10}}.
To extract such group of nodes of a network, one typically selects an objective function that captures the intuition of a community as a set of nodes with better internal connectivity than \mbox{external \cite{NewGir04,Luxburg07}}. 
The objective is generally NP-hard to optimize \cite{Schaeffer200727, NewGir04}; 
heuristics \cite{danon2005, CosciaGP11} or approximation \mbox{algorithms \cite{Schaeffer200727}} are used in practice to find sets of nodes that approximately optimize the objective function, which is interpreted as real communities.

 Another important approach is to define communities as the output of an algorithm that converges automatically, with some intuitive hope to extract good communities \cite{LeskovecLM10, raghavan}. 
Identified communities have some different importance in different domains.
In social networks, community means an organizational unit, in a biochemical network, a functional unit, in a collaboration network, a scientific discipline, and so on \cite{YangL12}.

Our observations regarding the development of network community detection algorithms are as follows:
(1) the network community detection is not easy NP-hard, like data clustering, due to the lack of good heuristics;
(2) both graph traversal-based methods and spectral methods are computationally overloaded due to the verification of the objective function value, which is required to guide the next iteration and;
(3) the rich literature of clustering is not very suitable for graph data. 

Some methods are available for network community detection, which tries to develop a similarity or distance function among the nodes of a complex network and to use that similarity or distance for partitioning the 
network \cite{Pons04, duch-2005, Chakrabarti04a, MacropolS10, LevoratoP11, BrandesGW03,bullmore2009complex}.
Most of the methods of community detection, based on similarity or distance, mainly use the shortest path, 
Jaccard similarity, set similarity or Euclidean distance, and they are less successful for network community detection in terms of conductance and modularity. In some cases, weighted graph are a requirement, which is not always obtained naturally in real networks. 
Complex networks are characterized by a small average path length and a high clustering coefficient; the way the metric is defined should be able to capture the crucial properties of complex networks. Therefore, we need to create the metric very carefully, so that it can explore the underlying community structure of the real-life networks. 

In this work, we develop the notion of a metric among the nodes using some new matrices derived from the modified adjacency matrix of the graph, which is flexible over the networks and can be tuned to enhance the structural properties of the network required for community detection. 
The main contributions of this work include: 
\begin{itemize}
\item A detailed study of the community detection algorithms. 
\item Transforming a graph to a metric space, preserving its structural properties.
\item Studying the complex properties of real-world networks on induced metric space.
\item Developing community detection algorithms on induced metric space.
\item Analyzing the results and complexities of the developed algorithms. 
\item Comparing the community detection algorithms with other existing methods.
\end{itemize}

The rest of this paper is organized as follows: Section \ref{gtm-rel} describes the state of the art of the network community detection literature. In Section \ref{gtm-met}, the problem of transforming a graph into a metric space is discussed, and the properties of a real complex network are studied.
In \mbox{Section \ref{gtm-comm}}, the problem of network community detection is formulated, and several possible solutions are presented in the induced metric space. Furthermore, the initialization procedures, termination criteria and convergence are discussed in detail. 
The results of the comparison between community detection algorithms are illustrated in \mbox{Section \ref{gtm-res}}. The computational aspects of the proposed framework are also discussed in this section.

\section{Network Community Detection}
\label{gtm-rel}

Community detection in real networks aims to capture the structural organization of the network using the connectivity information as the input \cite{NewGir04,Schaeffer200727}.
Early work on this domain was attempted by Weiss and
Jacobson while searching for a work group within a government \mbox{agency \cite{weiss}}.

Most of the methods developed for network community detection are based on a two-step approach.
The first step is specifying a quality measure (evaluation measure, objective function) that quantifies the desired properties of communities, and the second step is applying algorithmic techniques to assign the nodes of a graph into communities by optimizing the objective function.

Several measures for quantifying the quality of communities have been proposed; they mostly consider that communities are a set of nodes with many edges between them and few connections with nodes of different communities. Some of the community evaluation measures are described in the next subsection.

\subsection{Community Evaluation}

Several measures for quantifying the quality of communities have been proposed:

\begin{itemize}
\item Modularity: The notion of modularity is the most popular for network community detection purposes. 
The modularity index assigns high scores to communities whose internal edges are more than that expected in a random network model, which preserves the degree distribution of the given network.

\item Internal density: Density is defined by the number of edges ($m_s$) in subset $S$ divided by the total number of possible edges between all nodes $(n_s(n_s-1)/2)$. The ``$2$'' is there to cancel out duplicated edges. Internal density $= m_s/(n_s(n_s-1)/2)$

\item Edges inside: This is somewhat useless by itself, since it is not related to any other attributes of subset $S$; the total number of edges ($m_s$) present in subset $S$. Edges inside $= m_s$

\item Average degree: This is the average internal degree across all nodes ($n_s$) in subset $S$. 
Average degree $= 2m_s/n_s$

\item The fraction over the median degree: This determines the number of nodes that have an internal degree greater than the median degree of nodes in subset $S$.

\item Triangle Participation Ratio: The best measure for density, cohesiveness, and clustering within the goodness scales. Robust under random and expand perturbations. The fraction of nodes in $S$ that belong to a triad. TPR = (number of nodes belonging to a triad)$/n$.

\item Expansion: This measure of separability gives the average number of external connections ($c_s$) per node ($n_s$) in subset $S$ with graph $G$. It can be thought of as the external degree. Expansion $= c_s/(n_s(n-n_s))$.

\item Cut ratio: This metric is a measure of separability and can be thought of as external density. It is the fraction of external edges ($c_s$) of subset S
 out of the total number of possible edges in graph $G$.

\item Conductance: This is the ratio of edges inside the cluster to the number of edges leaving the cluster (captures the surface area to volume ratio). It measures best in separability (goodness scale), measuring well-separated non-overlapping communities. It is robust under node swap and shrink perturbation. Community-like sets of nodes have lower conductance.

\item Normalized cut: This represents how well subset $S$ is separated from graph $G$. It sums up the fraction of external edges over all edges in subset $S$ (conductance) with the fraction of external edges over all non-community edges.

\item Maximum out degree fraction: This metric first finds the fraction of external connections to internal connections for each node ($n_s$) in $S$. It then returns the fraction with the highest value.

\item Average out degree fraction: This is the sum of the individual fraction of edges outside of the community over the total connections of a node in subset $S$. It is then divided by the total number of nodes ($n_s$) in subset $S$.

\item Flake out degree fraction: This is a fraction of the number of nodes that have fewer internal connections than external connections to the number of nodes ($n_s$) in subset $S$.
\end{itemize}

There are several other measures of quality determination for a network community. However, the most widely-used measures are modularity and conductance. The majority of the algorithms are developed using either of the measures as their optimization criteria.

\subsection{Popular Algorithms}

In this subsection, we give a brief list of the algorithms developed for network community detection purposes. The basic approach and the complexity of execution is also given briefly (Table \ref{cnalgo}) in this subsection. 
\begin{itemize}
\item Fast greedy algorithm: This algorithm was developed by Newman {\textit{et al.}} \cite{newman03fast, Clauset2004}. It is modularity based and uses a hierarchical agglomerative approach. It is called fast greedy, because it is significantly faster than older algorithms and uses a greedy method.

\item Walktrap algorithm: This algorithm by Pons and Latapy \cite{Pons04} uses a hierarchical agglomerative method. Here, the distance between two nodes is defined in terms of a random walk process. The basic idea is that if two nodes are in the same community, the probability to get to a third node located in the same community through a random walk should not be very different. The distance is constructed by summing these differences over all nodes, with a correction for degree.

\item Eigenvector algorithm: This algorithm by Newman \cite{Newman:2006} is modularity based, and it uses an optimization method inspired by graph partitioning techniques. It relies on the eigenvectors of a so-called modularity matrix, instead of the graph Laplacian traditionally used in \mbox{graph partitioning}.

\item Label propagation algorithm: This algorithm by Raghavan \textit{et al.} \cite{raghavan} uses the concept of node neighborhood and the diffusion of information in the network to identify communities. Initially, each node is labeled with a unique value. Then, an iterative process takes place, where each node takes the label that is the most spread in its neighborhood. This process goes on until one of several conditions is met, for instance no label change. The resulting communities are defined by the last label values.

\item Spinglass algorithm: This algorithm by Reichardt and Bornholdt \cite{spinglass} is an optimization method relying on an analogy between the statistical mechanics of complex networks and physical spinglass models
\end{itemize}

There are more algorithms developed to solve the network community detection problem; a complete list can be obtained in several survey articles \cite{Fortunato10, LeskovecLM10, YangL12}. Some interesting recent articles are \cite{Deritei14, Mina09, pan14, Lee13, alde13, DeMeo14, Moustafa13}.

A partial list of algorithms developed for network community detection purpose is tabulated in 
Table \ref{cnalgo}. The algorithms are categorized into three main groups as spectral (SP), graph traversal based (GT) and semi-definite programming based (SDP). The categories and complexities are also given in the Table \ref{cnalgo}. 
\begin{table}[H]
\caption{Algorithms for network community detection and their complexities. GT, graph traversal; SDP, semi-definite programming; SP, spectral.}
\label{cnalgo}
\centering
\tiny
\begin{tabular}{lcll} 
\toprule 
{\bf Author} & {\bf Ref.} & {\bf Cat.}
{\bf  (No.)} & {\bf Order} \\[0.5ex]\hline
Van Dongen 	& (Graph clustering, 2000 \cite{dongen}) & GT(1) & $O(nk^2 )$, $k < n$ parameter \\
Eckmann and Moses & (Curvature, 2002 \cite{Eckmann2002}) & GT(2) & $O(m k^2 )$ \\
Girvan and Newman & (Modularity, 2002 \cite{girvan02}) & SDP(1) & $O(n^2 m)$ \\
Zhou and Lipowsky & (Vertex proximity, 2004 \cite{ZhouL04}) & GT(3) & $O(n^3 )$ \\
Reichardt and Bornholdt & (Spinglass, 2004 \cite{spinglass}) & SDP(2) & parameter dependent \\
Clauset \textit{et al.} 	& (Fast greedy, 2004 \cite{Clauset2004}) & SDP(3) & $O(n log_2 n)$ \\
Newman and Girvan 	& (Eigenvector, 2004 \cite{NewGir04}) & SP(1) & $O(nm^2 )$ \\
Wu and Huberman & (Linear time, 2004 \cite{linear}) 		& GT(4) & $O(n + m)$\\
Fortunato \textit{et al.} & (Infocentrality, 2004 \cite{infocentrality}) & SDP & $O(m^3 n)$ \\
Radicchi \textit{et al.} & (Radicchi \textit{et al.}, 2004 \cite{radicchi2004}) & SP(2) & $O(m^4 /n^2 )$\\
Donetti and Munoz & (Donetti and Munoz, 2004 \cite{Donetti}) & SDP(4) & $O(n^3 )$ \\
Guimera \textit{et al.} & (Simulated annealing, 2004 \cite{Guimera04}) & SDP(5) & parameter dependent \\ 
Capocci \textit{et al.} & (Capocci \textit{et al.}, 2004 \cite{Capocci04}) & SP(3) & $O(n^2 )$ \\
Latapy and Pons 		& (Walktrap, 2004 \cite{Pons04}) & SP(4) & $O(n^3 )$ \\
Duch and Arenas & (Extremal optimization, 2005 \cite{duch05}) & GT(5) & $O(n^2 log n)$ \\
Bagrow and Bollt & (Local method, 2005 \cite{Bagrow05}) & SDP(6) & $O(n^3 )$ \\
Palla \textit{et al.} & (overlapping community, 2005 \cite{palla05}) & GT(6) & $O(exp(n))$ \\
Raghavan \textit{et al.} & (label propagation, 2007 \cite{raghavan}) & GT(7) & $O(n + m)$\\
Rosvall and Bergstrom & (Infomap, 2008 \cite{rosvall2008random}) & SP(5) & $O(m)$ \\
Ronhovde and Nussinov & (Multiresolution community, 2009 \cite{Ronhovde09}) & GT(8) & $O(m\beta log n)$, $\beta \approx 1.3$ \\[1ex]\bottomrule
\end{tabular}
\end{table}

\subsection{Observations and Motivations}
\label{motivation}

Community detection is an extensively studied research problem of network science. However, a good algorithm for a large real network is still in demand for research communities. Two major criteria to be satisfied by good algorithms are: (1) they must find a partition of the network that is optimal with respect to modularity or conductance; and (2) the algorithm should be computationally efficient on large networks. The notable pitfalls of the existing algorithms are that most of the algorithms developed based on spectral methods or semi-definite programming rely on global optimization and need to compute the costlier functions under the evaluation criteria in each iteration and increase the burden of computation drastically, thus becoming inefficient for large networks. On the other hand, graph-based algorithms rely on local heuristic method or exhaustive search. The algorithms based on exhaustive search are not suitable for large networks. However, the local methods are computationally good, but fail to achieve a close value from the optimal modularity for large networks.  

A good alternative is to transform a network to a metric space, where we can achieve good optimality along with automatic convergence, thus leading to less computational burden for large networks; but, we need to create the metric very carefully, so that it can explore the underlying community structure of the real-life networks.

\section{Graph to Metric Space Transformation}
\label{gtm-met}

In this section, we demonstrate the procedure to transform a graph into points of a metric space and develop the methods of community detection with the help of a metric defined for a pair of points. We have also studied and analyzed the community structure of the network therein. 

As discussed in sub-section \ref{motivation}, the nodes of the graph do not lie on a metric space, e.g., edges do not reflect the Euclidean distance between the nodes. The standard Euclidean distance and spherical distance defined over the adjacency or Laplacian matrices above failed to capture similarity information among the nodes of a complex network. On the other hand, the algorithms developed based on the shortest path or Jaccard similarity are computationally inefficient and have less success in terms of standard evaluation criteria (like conductance and modularity). 

In this work, we have tried to develop the notion of similarity among the nodes using some new matrices derived from the adjacency matrix and the degree matrix of the graph.
Let $A$ be the adjacency matrix and $D$ the degree matrix of the graph $G= (V, E)$. The Laplacian $L = D -A$. We have defined two diagonal matrices of the same size $D (\lambda)$ and $D (\lambda_{x})$, where $\lambda$ is a parameter determined from the given graph and can be optimized from the optimization criteria of the problem under consideration.
In $D(\lambda)$, a fixed optimally-determined value is used in the diagonal entries of the matrix $D$, and in $D(\lambda_{x})$, a variable value, also optimally determined, is used in the diagonal entries of the matrix $D$.
The similarities are defined on matrices $L_1$ and $L_2$, where $L_1 = D(\lambda) +A$ and \mbox{$L_2 = D(\lambda_{x}) +A$}, respectively, are the spherical similarity among the rows and determined by applying a concave function $\phi$ over the standard notions of similarities, like the Pearson coefficient ($\sigma_{PC}$), the Spacerman coefficient ($\sigma_{SC}$) or the cosine similarity ($\sigma_{CS}$). $\phi(\sigma)()$ must be chosen using the chord condition to obtain a metric.


\subsection{Graph to Metric Space Algorithm}
\label{metric}

In this subsection, we demonstrate the algorithm to convert the nodes of the graph to the points of a metric space preserving the community structure of the graph.
The algorithm depends on the sub-modules (1) construction of $L_x$ ($L_1$ or $L_2$) and (2) obtaining a structure-preserving distance function. The algorithm works by picking a pair of nodes from $L_x$ and computing the distance defined in the second module.

\subsubsection{$L_x$ Construction} 

The $L_1$ is defined as $L_1 = D(\lambda) +A$, where $A$ is the adjacency matrix of the given network and $D(\lambda)$ is a diagonal matrix of the same size with diagonal values equal to a non-negative \mbox{constant $\lambda$}.

The $L_2$ is defined as $L_2 = D(\lambda_{x}) +A$, where $A$ is the adjacency matrix of the given network and $D(\lambda_{x})$ is a diagonal matrix of the same size with diagonal values determined by a non-negative function
 $\lambda_{x}$ of the node $x$.

The choice of $\lambda$ and $\lambda_{x}$ plays a crucial role in combination with the function chosen in the second module for the determination of a suitable metric and is discussed later in this subsection.

\subsubsection{Function Selection} 

The function selection module determines the metric for a pair of nodes. 
The function selector $\phi$ converts a similarity function (Pearson coefficient ($\sigma_{PC}$), Spacerman coefficient ($\sigma_{SC}$) or cosine similarity ($\sigma_{CS}$)) into a distance matrix. In general, the similarity function satisfies the positivity and similarity condition of the metric, but not the triangle inequality.
$\phi$ is a metric-preserving ($\phi(d(x_i,x_j)= d_{\phi}(x_i,x_j)$), concave and monotonically-increasing function. The three conditions above are referred to as the chord condition. The $\phi$ function is chosen to have minimum internal area \mbox{with the chord}.

\subsubsection{Choice of $\lambda$ and $\phi(\sigma)()$} 

The choices in the above sub-modules play a crucial role in the graph to metric transformation algorithm to be used for community detection. The complex network is characterized by a small average diameter and a high clustering coefficient. Several studies on network structure analysis reveal that there are hub nodes and local nodes characterizing the interesting structure of the complex network.
 Suppose we have taken $\phi= arccos$, $\sigma_{CS}$ and constant $\lambda \geq 0$. $\lambda = 0$ penalizes the effect of the direct edge in the metric and is suitable to extract communities from a highly dense graph. 
$\lambda = 1$ places a similar weight of the direct edge, and the common neighbor reduces the effect of the direct edge in the metric and is suitable to extract communities from a moderately dense graph. $\lambda = 2$ sets more importance for the direct edge than the common neighbor (this is the common case of available real networks). 
$\lambda \geq 2$ penalizes the effect of the common neighbor in the metric and is suitable for extracting communities from a very sparse graph. 

The choice of $\lambda$ depends on the data complexity for community detection (DCC) value (sub-section \ref{complexity}) of the input graph, {\em i.e.}, whether it is sparse or dense, and its cluster structure. 
 
The algorithm for transforming a graph to the points of a metric space is given in Algorithm \ref{alg1}.
 
  \begin{theorem}

 $M =(V, d)$ constructed in the above Algorithm \ref{alg1} is a metric space with respect to the metric $d$, \textit{i.e.},:

\end{theorem}
 The proof of the theorem is straight forward and satisfies the following metric properties:         
\begin{itemize}
\item $d(v_i,v_j) \geq 0$ $and$ $d(v_i,v_i)= 0$
\item $d(v_i,v_j) = d(v_j,v_i)$ 
\item $d(v_i,v_j) \leq d(v_i,v_k) + d(v_k,v_j)$
\end{itemize}   
             
\begin{algorithm}[H]
\begin{algorithmic}[1]
\REQUIRE $G =(V, E)$
\ENSURE $M= (V, d)$ 
\STATE $D^{\lambda}_{(n \times n)} = \begin{cases} 0 ~if~ i \neq j \\\lambda \geq 0 ~if~ i= j \end{cases}$ 
\STATE $A= D^{\lambda} +E$
\FOR{$i=1$ to $n$}
 \FOR{$j=1$ to $n$}
\STATE $d(v_i, v_j) =\phi( 1- \frac{a_i \cdot a_j}{|a_i||a_j|})$, where $v_i, v_j \in V$ and $a_k$ is the $k$-th row of $A$ and $\phi$ is an affine function.
\ENDFOR
\ENDFOR
\RETURN $M= (V, d)$
\end{algorithmic}
\caption{Mapping a graph into the metric space.}
\label{alg1}
\end{algorithm}


\section{Community Detection on Induced Metric Space}
\label{gtm-comm}

In this section, we explore the k partitioning algorithm for the purpose of network community detection by using the metric space constructed above for each graph. We have also studied and analyzed the advantages of the k partitioning method over the standard algorithm for network community detection.

\subsection{k-Partitioning}
\label{partition}

The community detection methods based on {\em k}-partitioning of a graph are possible using the newly-defined node distance, because the nodes of the graph are converted into the points of a metric space. The {\em k}-partitioning of a graph uses this distance converges automatically and does not compute the value of objective function in iterations; therefore, it reduces the computation compared to standard graph partitioning methods. The results of {\em k}-partitioning of a graph using a metric are competitive on the large set of networks shown in Section \ref{gtm-res}. The algorithm for community detection using {\em k}-partitioning and its detailed analysis is given below (Algorithm \ref{alg2}). Before that, we need to determine the value of {\em k}, and that is discussed in the next sub-section. 

\subsection{k Selection}
\label{k-selection} 

Determining the optimal number of {\em k} is an important problem for community detection researchers. An extensive analysis can be found in the work of Leskovec {\em \textit{et al.}} \cite{Leskovec2008}. The standard practice is to solve an optimization equation with respect to {\em k} for which the optimal value of the objective function is achieved. Another method based on farthest first traversal is also very useful in terms of computational efficiency \cite{Gonzalez85}. For small networks, the global optimization works better, and for a very large network, the second choice gives a faster approximate solution.

\subsection{Initialization for k-Partitioning} 
The set of initial nodes are also a very important problem for the {\em k} partitioning algorithm:
\begin{itemize}
\item Input: graph $G = (V,E)$, with the node similarity $sim(x_a,x_b)$ defined on it,
\item Output: A partition of the nodes into $k$ communities $C_1, C_2,..., C_k $,
\item Objective function: Maximize the minimum intra-community similarity: 
$$min_{j\in\{1,2,..,k\}}max_{x_a,x_b \in C_j} ~~ sim(x_a,x_b)$$ 

\end{itemize}

\begin{algorithm}[H]
\begin{algorithmic}[1]
\REQUIRE $M= (V, d)$ 
\ENSURE $T= \{C_1, C_2, \dots, C_k \}$ with minimum $cost(T)$
\STATE Initialize centers $z_1, \dots, z_k \in R^n$ and clusters $T= \{C_1, C_2, \dots, C_k\}$ 
\REPEAT 
	\FOR{$i=1$ to $k$}	
		\FOR{$j=1$ to $k$}
			\STATE $C_i \leftarrow \{x \in V ~s.t.~ |z_i-x| \leq |z_j-x|\}$
		\ENDFOR
	\ENDFOR 
	\FOR{$j=1$ to $k$}	
		\STATE $z_i \leftarrow mean(C_i)$
	\ENDFOR 
\UNTIL $|cost(T_t) - cost(T_{t+1})| = 0$
\RETURN $T= \{C_1, C_2,\dots, C_k \}$

\end{algorithmic}
\caption{{\em k}-center partitioning algorithm.}
\label{alg2}
\end{algorithm}

\subsection{Convergence}
\label{convergence} 

Convergence of the network community detection algorithms is the least studied research area of network science. However, the rate of convergence is an important issue, and a low rate of convergence is the major pitfall of most of the existing algorithms. Due to the transformation into the metric space, our algorithm is equipped with the quick convergence facility of the k-partitioning on the metric space by providing a good set of initial points. Another crucial pitfall suffered by the majority of the existing algorithms is the validation of the objective function in each iteration during convergence. Our algorithm converges automatically to the optimal partition, thus reducing the cost of validation during convergence. 

\begin{theorem}
During the course of the $k$ center partitioning algorithm, the cost monotonically decreases.
\end{theorem}

\begin{proof}
Let $Z^t = \{z_1^t ,\dots , z_k^t\}$ , $T^t = \{ C_1^t ,\dots, C_k^t\}$ denote the centers and clusters at the start of the $t$-th iteration of the $k$ partitioning algorithm. The first step of the iteration assigns each data point to its closest center; therefore, $cost(T^{t+1},Z^t) \leq cost(T^t,Z^t) $.

In the second step, each cluster is re-centered at its mean; therefore, $cost(T^{t+1},Z^{t+1}) \leq cost(T^{t+1},Z^t)$. 

\end{proof}

\begin{theorem}
If $T$ is the solution returned by farthest-first traversal and $T^{o}$ is the optimal solution, then $cost(T^{o}) \leq cost(T) \leq 2cost(T^{o}) $.
\end{theorem}
\begin{proof}
The proof of the theorem can be obtained in \cite{Gonzalez85}.

\end{proof}

\subsection{Data Complexity}
\label{complexity}

The key characteristics of complex network are ``high clustering coefficient'' and ``small average path length''.
The first property justifies the community structure of the network, whereas the second property justifies the small world phenomena of real networks.
Given a network, that is given a number of nodes and a number of edges, what are the bounds of the average distance and clustering coefficient? The two properties of the optimal complex network (OCN) are (1) the minimum possible average distance and (2) the maximum possible clustering coefficient.
There is usually a unique graph with the largest average clustering, which at the same time has the smallest possible average distance.
In contrast, there are many graphs with the same minimum average distance, ignoring their average clustering.
The objective of this work is to measure the community detectability of the complex network, 
$G(N, m, L, C)$, where $N$ is the number of vertices, $m$ is the number of edges, $L$ is the average path length and $C$ is the average clustering coefficient.

Average path length: $L_{N,m}$.
The smallest possible average distance
of a graph with $N$ vertices and $m$ edges we denote $L_{N,m} = \frac{1}{m}\sum_{u,v \in E} d(u,v)$.

Clustering coefficient: If $d_u(> 1)$ is the degree of a vertex $u$ and $t_u$ is the number of edges among its neighbors, its clustering coefficient is
 $C(u) = t_u/ \left(\begin{array}{c} d_u \\ 2 \end{array}\right)$.

In some graphs, community detection is easy, and most of the algorithms work very well (e.g., disjoint cliques). On the other hand, in some graphs, community detection is very difficult, and some algorithms rarely work well (e.g., circular graph).

Data complexity of community detection: Informally, Given a graph with $N$ vertices and $m$ edges $G(N, m)$, to what extent we can reveal the community structure is the data complexity for community detection of that graph. Data complexity for community detection (DCC) is denoted as ($\alpha(G(N,m, L,C))$), $\alpha(G(N,m, L,C))$ near zero for a graph for which is is easy to detect community and $\alpha(G(N,m, L,C))$ near one with no community structure. 
DCC is calculated as the ratio between common edges of $G^{*}(N,m,L,C)$ and $G(N,m, L,C)$ with $m$ the number of edges of $G$ or $G^{*}$, where $G^{*}(N,m,L,C)$ is a graph with the same average path length constructed by adding the minimum number of edges to an empty graph of $N$ nodes followed by the addition of more edges to obtain the total number $m$ by maximizing the clustering coefficient.

A higher value of DCC for a particular network signifies that we can extract a good community structure of the network; however, a lower value of DCC signifies that none of the algorithms are very useful to capture the community structure of the network. Another advantage of DCC is that it can assess the quality of an algorithm. When DCC is high and the value of the evaluation measure is low, it simply signifies that there is enough room to improve the algorithm.

\section{Experiments and Results}
\label{gtm-res}

We performed many experiments to test the proposed network detection method via induced metric space over several real networks given in Table \ref{cndata}. The objective of the experiment is to verify the behavior of the algorithm and the time required to compute the algorithm. One of the major goals of the experiment is to see the behavior of the algorithm with respect to the change of values of the crucial limits of the data and the parameters of the algorithm.

Experiments are also conducted to compare the results (Tables \ref{cnres}, \ref{cnresmod} and \ref{cnrestime}) of our algorithm with the state-of-the art-algorithms (Table \ref{cnalgo}) available in the literature in terms of common measures mostly used by the researchers of the domain of network community detection. The details of several experiments and the analysis of the results are given in the following subsections. 

\subsection{Experimental Designs}

Experiment for comparison: In this experiment, we compared several algorithms for network community detection with our proposed algorithm based on metric space. The experiment is performed on a large list of network datasets. Two versions of the experiment are developed for comparison purposes based on two different quality measures: conductance and modularity. The results are shown in the Tables \ref{cnres} and \ref{cnresmod}, respectively.

Experiment on the performance and time: In this experiment, we evaluated our algorithm for the performance on the network collection (Table \ref{cndata}). We evaluated the time taken by our algorithm on different sizes of networks, and this is shown in the Table \ref{cnrestime}.


\subsection{Performance Indicator}

Modularity: The notion of modularity is the most popular for network community detection purposes. 
The modularity index assigns high scores to communities whose internal edges are more than expected in a random network model, which preserves the degree distribution of the \mbox {given network}.

Conductance: Conductance is widely used in the graph partitioning literature. 
The conductance of a set $S$ with complement $S^{C}$ is the ratio of the number of edges connecting nodes in $S$ to nodes in $S^{C}$ by the total number of edges incident to $S$ or to $S^{C}$ (whichever number is smaller).

\subsection{Datasets}

A list of real networks taken from several real-life interactions is considered for our experiments, and they are in Table \ref{cndata} below. We have also listed the number of nodes, the number of edges, the average diameter, the data complexity for community detection (DCC) and the {\em k} value used (sub-section \ref{k-selection}). The values of the last column can be used to assess the quality of detected communities, as discussed in the sub-section \ref{complexity}. 

\begin{table}[H]
\caption{Complex network datasets and the values of their parameters. DCC, data complexity for community detection.}
\label{cndata}
\centering
\tiny
\begin{tabular}{l c l l c c c} 
\toprule

 {\bf Name} & {\bf Type} & {\bf No. of Nodes} & {\bf No. of Edges} & {\bf Diameter} & {\bf DCC} &  {\bf k} \\ [0.5ex] 
 \hline
Facebook 	& U & 4039 	& 88,234	& 4.7 & 0.72498 & 164 \\	
Gplus 		& D & 107,614 	& 13,673,453 	& 3 & 0.50073 & 457 \\
Twitter		& D & 81,306 	& 1,768,149 	& 4.5 & 0.57072 & 213 \\	
Epinions1 	& D & 75,879 	& 508,837 	& 5 & 0.14001 & 128 \\
LiveJournal1 	& D & 4,847,571 & 68,993,773 	& 6.5 & 0.27432 & 117 \\	
Pokec 		& D & 1,632,803 & 30,622,564 	& 5.2 & 0.10971 & 246 \\	
Slashdot0811 	& D & 77,360 	& 905,468 	& 4.7 & 0.05884 & 81 \\	
Slashdot0922 	& D & 82,168 	& 948,464 	& 4.7 & 0.06340 & 87 \\	
Friendster	& U & 65,608,366& 1,806,067,135 & 5.8 & 0.16231 & 833 \\
Orkut		& U & 3,072,441	& 117,185,083 	& 4.8 & 0.16689 & 756 \\	
Youtube		& U & 1,134,890 & 2,987,624 	& 6.5 & 0.08090 & 811 \\	
DBLP	 	& U & 317,080	& 1,049,866 	& 8 & 0.63307 & 268 \\
Arxiv-AstroPh	& U & 18,772	& 396,160 	& 5 & 0.65841 & 23 \\	
web-Stanford	& D & 281,903 	& 2,312,497 	& 9.7 & 0.60034 & 69 \\	
Amazon0601	& D & 403,394	& 3,387,388	& 7.6 & 0.41890 & 92 \\	
P2P-Gnutella31	& D & 62,586	& 147,892 	& 6.5 & 0.00710 & 35 \\	
RoadNet-CA	& U & 1,965,206 & 5,533,214 	& 500 & 0.40458 & 322 \\	
Wiki-Vote 	& D & 7115 	& 103,689 	& 3.8 & 0.17048 & 21 \\ [1ex] 	
\bottomrule
\end{tabular}
\end{table}

\subsection{Computational Results}

In this subsection, we compare two groups of algorithms for network community detection with our proposed algorithm based on metric space. The experiment is performed on a large list of network datasets. Two versions of the experiment are developed for comparison purposes based on two different quality measures: conductance and modularity. The results based on conductance are shown in the \mbox{Table \ref{cnres}}, and the results based on modularity are shown in the Table \ref{cnresmod}, respectively. Regarding the two groups of algorithms, the first group contains algorithms based on semi-definite programming, and the second group contains algorithms based on graph traversal approaches. For each group, we have taken the best value of conductance in Table \ref{cnres} and the best value of modularity in Table \ref{cnresmod} among all of the algorithms in the groups. The results obtained with our approach are very competitive with most of the well-known algorithms in the literature, and this is justified over the large collection of datasets. On the other hand, it can be observed that time taken (Table \ref{cnrestime}) by our algorithm is quite less compared to other methods and justifies the theoretical findings described in Sections \ref{gtm-met} and \ref{gtm-comm}.

\begin{table}[H]
\caption{Comparison of our approaches with other best methods in terms of conductance; the numbers inside the brackets denote the algorithm of the group.}
\label{cnres}
\centering
\tiny
\begin{tabular}{l c c c c} 
 \toprule
 {\bf Name} 		& {\bf Spectral} 	& {\bf SDP} 	 & {\bf GT}  & {\bf Metric} \\ [0.5ex] 
 \hline
Facebook 	&  0.0097(5) & 0.1074(3) & 0.1044(7) & 0.1082 \\
Gplus  	&  0.0119(5) & 0.1593(3) & 0.1544(7) & 0.1602 \\
Twitter 	&  0.0035(5) & 0.0480(3) & 0.0465(7) & 0.0483 \\
Epinions1 	&  0.0087(5) & 0.1247(6) & 0.1208(7) & 0.1254 \\
LiveJournal1 	&  0.0039(5) & 0.0703(6) & 0.0680(7) & 0.0706 \\
Pokec  	&  0.0009(4) & 0.0174(3) & 0.0168(7) & 0.0175 \\
Slashdot0811 	&  0.0005(5) & 0.0097(6) & 0.0094(7) & 0.0098 \\
Slashdot0922 	&  0.0007(4) & 0.0138(3) & 0.0133(5) & 0.0138 \\
Friendster 	&  0.0012(5) & 0.0273(1) & 0.0263(7) & 0.0273 \\
Orkut 	&  0.0016(5) & 0.0411(3) & 0.0397(7) & 0.0412 \\
Youtube 	&  0.0031(5) & 0.0869(3) & 0.0838(7) & 0.0871 \\
DBLP 		&  0.0007(4) & 0.0210(3) & 0.0203(7) & 0.0211 \\
Arxiv-AstroPh 	&  0.0024(5) & 0.0929(6) & 0.0895(7) & 0.0931 \\
web-Stanford 	&  0.0007(5) & 0.0320(1) & 0.0308(7) & 0.0320 \\
Amazon0601 	&  0.0018(5) & 0.0899(6) & 0.0865(7) & 0.0900 \\
P2P-Gnutella31 &  0.0009(5) & 0.0522(6) & 0.0503(7) & 0.0523 \\
RoadNet-CA 	&  0.0024(5) & 0.1502(3) & 0.1445(7) & 0.1504 \\
Wiki-Vote 	&  0.0026(5) & 0.1853(6) & 0.1783(7) & 0.1855 \\[1ex] 
 \bottomrule
\end{tabular}
\end{table}

\begin{table}[H]

\caption{Comparison of our approaches with other best methods in terms of modularity; the numbers inside the brackets denote the algorithm of the group.}
\label{cnresmod}
\centering
\tiny
\begin{tabular}{l c c c c} 
 \toprule
 {\bf Name} 		& {\bf Spectral} & {\bf SDP} 	& {\bf GT} 	 & {\bf Metric} \\ [0.5ex] 
 \hline
Facebook 	& 0.4487(1) & 0.5464(4) & 0.5434(5) & 0.5472 \\
Gplus 	& 0.2573(1) & 0.4047(3) & 0.3998(5) & 0.4056 \\
Twitter 	& 0.3261(3) & 0.3706(1) & 0.3691(7) & 0.3709 \\
Epinions1 	& 0.0280(1) & 0.1440(3) & 0.1401(5) & 0.1447 \\
LiveJournal1 	& 0.0791(1) & 0.1455(5) & 0.1432(5) & 0.1458 \\
Pokec 	& 0.0129(3) & 0.0294(1) & 0.0288(5) & 0.0295 \\
Slashdot0811 	& 0.0038(1) & 0.0130(4) & 0.0127(7) & 0.0131 \\
Slashdot0922 	& 0.0045(1) & 0.0176(5) & 0.0171(5) & 0.0176 \\
Friendster 	& 0.0275(4) & 0.0536(5) & 0.0526(7) & 0.0536 \\
Orkut 	& 0.0294(3) & 0.0689(4) & 0.0675(5) & 0.0690 \\
Youtube 	& 0.0096(1) & 0.0934(2) & 0.0903(5) & 0.0936 \\
DBLP 		& 0.4011(5) & 0.4214(1) & 0.4207(5) & 0.4215 \\
Arxiv-AstroPh & 0.4174(3) & 0.5079(3) & 0.5045(5) & 0.5081 \\
web-Stanford 	& 0.3595(5) & 0.3908(4) & 0.3896(7) & 0.3908 \\
Amazon0601 	& 0.1768(1) & 0.2649(4) & 0.2615(7) & 0.2650 \\
P2P-Gnutella31 & 0.0009(1) & 0.0522(2) & 0.0503(5) & 0.0523 \\
RoadNet-CA 	& 0.0212(3) & 0.1690(4) & 0.1633(5) & 0.1692 \\
Wiki-Vote 	& 0.0266(1) & 0.2093(1) & 0.2023(5) & 0.2095 \\[1ex] 
 \bottomrule
\end{tabular}
\end{table}

\begin{table}[H]

\caption{Comparison of our approaches with other best methods in terms of time.}
\label{cnrestime}
\centering
\tiny
\begin{tabular}{l c c c c } 
 \toprule
 {\bf Algorithm}	& {\bf Spectral} 	& {\bf SDP} 	& {\bf GT} 		& {\bf Metric} \\ [0.5ex] \hline
Minimum Time 	& 884 		& 910 	& 871 		& 869 	\\
Maximum Time 	& 1386 	& 1725 & 1641 	& 869		\\
Average Time 	& 917 		& 981 	& 1338 	& 869 		\\ [1ex] 
 \bottomrule
\end{tabular}
\end{table}

\subsection{Parameter Settings} 

The values of several parameters are very crucial in our algorithm. Here, we discuss the different settings of $k$, $\lambda$, DCC and the affine function. For each datum described in Table \ref{cndata}, the $k$ value is obtained by optimizing the conductance value, as described in Subsection \ref{k-selection}, and the values are provided in \mbox{Table \ref{cndata}}. For small datasets (not considered for our experiments), the results are very sensitive to $k$, whereas for large networks (all of the above list), the results are less sensitive to $k$. The value $\lambda$ is taken $\lambda=2$ in all of the computation above; however, the results can be improved more by optimizing $lambda$. The DCC value provides us prior information about the community structure; it can be observed that we obtained good community structure where the DCC value is high. In all of the experiments described above, the $\phi(\sigma)()$ is constructed with the arccos function and cosine similarity.  

\subsection{Results Analysis and Achievements}

In this subsection, we describe the analysis of the results obtained in our experiments shown above and also highlight the achievements from the results. It is clearly evident from the results shown in Tables \ref{cnres}, \ref{cnresmod} and \ref{cnrestime} that the proposed metric-based method for network community detection provides very good competitive performance with respect to conductance modularity and time. However, a good community detection algorithm must provide the results close to the unknown optimal community structure. To assess the optimality, we have considered the best results of each class of algorithms and treated them as one of the best known estimate to the optimal community structure of the network. It is also evident from the results that our method provides results very close to the considered estimates of optimal communities.


\section{Conclusions}

Network community detection became an important research problem in recent years. 
In this article, we have demonstrated and analyzed a new approach to network community detection via metric space induced by the graph. The main achievement of the work was to use the rich literature of clustering in metric space. Clustering is easy NP-hard in metric space, whereas network community detection is NP-hard. 
The results obtained with our approach were very competitive with most of the well-known algorithms in the literature and justified over the large collection of datasets. Our algorithm converges automatically to optimal clustering. It does not require verifying the objective function value to guide the next iteration, like popular approaches, thus saving the time of computation.


\acknowledgments{Acknowledgments}

This work is supported by the Jaypee University of Information Technology.


\authorcontributions{Author Contributions}

Suman Saha proposed the algorithm and prepared the manuscript. Satya P. Ghrera was in charge of the
overall research and critical revision of the paper.


\conflictofinterests{Conflicts of Interest}

The authors declare no conflict of interest. 

\bibliographystyle{mdpi}
\makeatletter
\renewcommand\@biblabel[1]{#1. }
\makeatother



%


%

\end{document}